\documentclass{article}
\usepackage{spconf,amsmath,amssymb,graphicx,hyperref}
\usepackage{amsthm}

\usepackage{stmaryrd}
\usepackage{cite}
\usepackage{caption}
\newtheorem{thm}{\protect\theoremname}
\newtheorem{prop}[thm]{\protect\propositionname}

\providecommand{\examplename}{Example}
\providecommand{\propositionname}{Proposition}
\providecommand{\theoremname}{Theorem}

\title{Set-membership localization of intermittent RF sources\\
using a fleet of collaborating UAVs}
\name{Jacques Bois, Raul de Lacerda, Michel Kieffer}
\address{Université Paris-Saclay, CNRS, CentraleSupélec, \\Laboratoire des Signaux et Systèmes, 91192, Gif-sur-Yvette, France.}
\begin{document}
\ninept
\maketitle
\begin{abstract}
This paper proposes a set-membership approach (SMA) to localize radio frequency (RF) sources observed by a collaborating fleet of Unmanned Aerial Vehicles (UAVs). Considering frequency-separable RF transmitters with intermittent and periodic emission patterns, 
the SMA evaluates set estimates of the source locations and a set free of sources. Simulation results show that SMA outperforms a Bayesian baseline approach in terms of localization accuracy and convergence speed.
\end{abstract}
\section{Introduction}

\label{sec:Introduction}

Localizing radio frequency (RF) sources with periodic emissions using a fleet of collaborating Unmanned Aerial Vehicles (UAVs) is critical in many applications, including post-disaster search and rescue~\cite{Yanmaz2023}, border and critical infrastructure surveillance~\cite{Jayachandran2022}, electromagnetic intelligence~\cite{Bisio2021}, and IoT asset monitoring~\cite{Delafontaine2020}. Sources with periodic transmissions are commonly encountered in civil aviation emergency beacons~\cite{RTCA2020}, cellular synchronization channels~\cite{Zhao2002}, GNSS repeaters~\cite{Abraha2016}, and dedicated tracking systems~\cite{Wang2023a}.

Signal detection is a prior to RF source localization. While matched filtering~\cite{Krim1996} or time-frequency approaches~\cite{Nair2022} perform properly when waveform signatures are well known, their performance degrades significantly under unknown parameter fluctuations~\cite{Tian2020,Rojas2026,Gutu2026}. More fundamentally, emission intermittency complicates signal association and source separation~\cite{Pack2009,Koohifar2018,Testi2022}. Standard probabilistic formulations either assign empirical emission probabilities and address localization via maximum likelihood schemes~\cite{Oispuu2010}, or introduce fixed missed-detection rates~\cite{Selvaratnam2017}. 

RF localization classically exploits Time Difference of Arrival ~\cite{Khalil2023}, Received Signal Strength Indicator (RSSI)~\cite{Ketabalian2024,Zhou2024}, Frequency Difference of Arrival ~\cite{Pei2023}, or Angle of Arrival ~\cite{Vrba2019,Cheng2020}. Measurement fusion via Bayesian frameworks~\cite{Koohifar2018,Jiang2021} or Probability Hypothesis Density filters~\cite{Dames2017} allows source localization and tracking under clutter. However, when deployed to localize intermittent sources, probabilistic approaches struggle to 
certify the absence of sources within an observed region. Absence of detection during an observation window may indicate either an empty zone or an emitter remaining in a silent period. State-machine heuristics~\cite{Pack2009} and Bernoulli-based intermittency tracking~\cite{Koohifar2018} provide confidence probabilities, yet they fail to deliver deterministic guarantees of clearance.

Deploying mobile UAV fleets mitigates spatial occlusion and enhances coverage through collaborative sensing. Multi-agent coordination requires robust information fusion and path planning, typically realized via distributed particle or Kalman filtering~\cite{Ma2020}, rule-based collision-avoiding flocking~\cite{Spyridis2021}, Fisher information and Cram\'er-Rao Lower Bound optimization~\cite{Shahidian2016,Koohifar2017,Golihaghighi2022}, distributed multi-agent optimization~\cite{Zhang2024}, or reinforcement learning~\cite{Tang2025}. 

As a robust alternative to stochastic filtering, set-membership approaches (SMA)~\cite{Milanese2004} are able to characterize sets 
guaranteed to contain all value of the uncertain parameters consistent with the measurement model and noise bounds. SMA have been applied for mobile network positioning~\cite{Calafiore2026}, RSSI filtering~\cite{Yang2022,Liu2026}, and inter-UAV cooperative localization~\cite{Ding2026}. While vision-based UAV setups have applied SMA to occluded sources~\cite{Ibenthal2021}, apparent intermittency in these contexts arises solely from line-of-sight visual masking ~\cite{Zagar2025} rather than genuine temporal RF emission characteristics.

To address these limitations, this paper introduces a distributed SMA for localizing static and frequency-separable RF sources with intermittent but periodic emissions using a fleet of collaborating UAVs. The remainder of this paper is structured as follows. Section~\ref{sec:Problem} details the system model and problem formulation. Section~\ref{sec:Distributed-estimation} details the iterative SMA for source localization and absence of source certification. Section~\ref{sec:Simu-results} compares the performance of with proposed SMA with a Bayesian baseline approach.

\section{Hypotheses and problem formulation}
\label{sec:Problem}



Consider a unknown number $N^{\text{s}}$ of static sources emitting RF signals and situated in a bounded $2$-D area $\mathbb{X}_{0}\subset\mathbb{R}^3$ to which a frame $\mathcal{F}$ is attached. The source locations are $\boldsymbol{x}^{\mathrm{t}}_{j}$,
$j\in\mathcal{N}^{\mathrm{t}}=\llbracket1,N^{\text{t}}\rrbracket$.
Source~$j$ emits periodically RF signals of duration $\delta_{j}$, with period $\tau_{j}\in[\delta_{j},\tau_{\max}]$, where $\tau_{\max}$
is the maximum emission period, and emission offset 
$t_{0,j}$, defined with respect to a global temporal reference. The emission duration, period, and offset are not known. The signals emitted by two sources are considered as separable.

A fleet of $N^{\text{u}}$ UAVs explores $\mathbb{X}_{0}$ to detect the sources and estimate their location. Time is slotted in intervals of duration $T$. 
The state of UAV~$i$ at time $kT$ is $\boldsymbol{x}^{\mathrm{u}}_{i}(k)$
and consists of the location and orientation of the origin of $\mathcal{F}_{i}$, the frame attached to UAV~$i$
in $\mathcal{F}$. We assume that each UAV knows its own state. 

Each UAV $i\in\mathcal{N}^{\mathrm{u}}=\llbracket1,N^{\text{u}}\rrbracket$
is equipped with an antenna and signal processing system. As in \cite{Mechitov2003}, we assume that during the time interval $[(k-d)T,kT]$, UAV~$i$ is only able to detect the presence of a source in its field
of detection (FoD) $\mathbb{F}_{i}\left(k\right)$, which shape and location depends on $\boldsymbol{x}^{\mathrm{u}}_{i}(k)$.
If a source~$j$
emits at time $t$ such that $[t,t+\delta_{j}]\cap[(k-1)T,kT]\neq\emptyset$
and $\boldsymbol{x}^{\mathrm{t}}_{j}\in\mathbb{F}_{i}\left(k\right)$,
then source~$j$ is always detected by UAV~$i$. 

The internal clocks of all UAVs of the fleet are assumed to be perfectly
synchronized. 
The communication latency and errors are neglected, while the communication range
$r_\mathrm{com}$ is assumed to be limited: the list of one-hop neighbors of UAV~$i$ during the time interval $[kT,(k+1)T]$ is
denoted $\mathcal{N}_{i}(k)$.

\subsection{Set estimates and problem formulation}

\label{subsec:Estimates}

The set $\mathbb{I}_{i}(k)$ gathers the information available to UAV~$i$ up to time $kT$, including $\boldsymbol{x}^{\mathrm{u}}_{i}$$\left(k\right)$, the list
$\mathcal{K}_{i}(k)$ of indexes of already known sources (detected
up to $kT$), all its measurements, and similar information obtained from its neighbors
via communications.

As will be detailed in Section~\ref{sec:Distributed-estimation}, during the time interval $[kT,(k+1)T]$, using $\mathbb{I}_{i}(k)$ and information obtained from its neighbours, UAV~$i$ evaluates a list $\mathcal{X}_{i}(k)=\left\{ \mathbb{X}_{i,j}(k)\right\} _{j\in\mathcal{K}_{i}(k)}$
of set estimates $\mathbb{X}_{i,j}(k)$ of $\boldsymbol{x}^{\mathrm{t}}_{j}$ for all sources with index in $\mathcal{K}_{i}(k)$ and
a set $\overline{\mathbb{X}}_{i}(k)$ containing all locations in $\mathbb{X}_{0}$ where UAV~$i$ can state that there is no source.
Then UAV~$i$ evaluates the source localization uncertainty
\begin{equation}
\Phi^{\mathrm{u}}_{i}(k)=\phi\left(\textstyle\bigcup_{j\in\mathcal{K}_{i}(k)}\mathbb{X}_{i,j}(k)\right),\label{eq:presence-uncertainty}
\end{equation}
where $\phi(\mathbb{X})$ is a measure of the set $\mathbb{X}$. UAV~$i$ also evaluates the size of the set potentially containing sources  
\begin{equation}
\Phi^{\mathrm{a}}_{i}(k)=\phi\left(\mathbb{X}_{0}/\overline{\mathbb{X}}_{i}(k)\right).\label{eq:absence-metric}
\end{equation}

The objective of the fleet of UAVs is to minimize
\begin{equation}
\Phi\left(k\right)=\sideset{\frac{1}{N^{\text{u}}}}{_{i=1}^{N^{\text{u}}}}\sum\left(\Phi^{\mathrm{u}}_{i}(k)+\lambda\Phi^{\mathrm{a}}_{i}(k)\right),\label{eq:uncertainty-estimate}
\end{equation}
where $\lambda\geqslant 0$ adjusts the trade-off between the source localization uncertainty and the uncertainty about the absence of sources.


\section{Iterative Set-Membership Estimation}

\label{sec:Distributed-estimation}


Assuming  that a given UAV~$i$, immediately before time $kT$,  has access only to $\mathcal{K}_{i}(k-1)$, $\mathcal{X}_{i}(k-1)=\left\{ \mathbb{X}_{i,j}(k-1)\right\} _{j\in\mathcal{K}_{i}(k-1)}$,
and $\overline{\mathbb{X}}_{i}(k-1)$, with these sets initialized
as $\mathcal{K}_{i}(0)=\emptyset$, $\mathcal{X}_{i}(0)=\emptyset$,
and $\overline{\mathbb{X}}_{i}(0)=\emptyset$, the list $\mathcal{K}_{i}(k)$ is updated at time $kT$ based on the list $\mathcal{D}_{i}(k)$, which contains the indexes of the sources detected by UAV~$i$ during the time interval $\left[\left(k-1\right)T,kT\right]$




\begin{equation}
\mathcal{K}_{i}(k)=\mathcal{K}_{i}(k-1)\cup\mathcal{D}_{i}(k).\label{eq:Updating list of detected sources}
\end{equation}
For a newly detected source, 
\emph{i.e.}, such that $j\in\mathcal{D}_{i}(k)$ and $j\notin\mathcal{K}_{i}(k-1)$,
the set estimate of the location $\boldsymbol{x}_j$ is evaluated as 
\begin{equation}
\mathbb{X}_{i,j}(k)=\mathbb{F}_{i}(k)\setminus\overline{\mathbb{X}}_{i}(k-1),\label{eq:Updating solution sets if new source}
\end{equation}
as $\overline{\mathbb{X}}_{i}(k-1)$ is free of sources. For a previously
detected source $j$ that is detected again, \emph{i.e.}, such that $j\in\mathcal{D}_{i}(k)\cap\mathcal{K}_{i}(k-1)$, 
\begin{equation}
\mathbb{X}_{i,j}(k)=\left(\mathbb{X}_{i,j}(k-1)\cap\mathbb{F}_{i}(k)\right)\setminus\overline{\mathbb{X}}_{i}(k-1),\label{eq:Updating solution sets}
\end{equation}
since the location of such source is known to belong to $\mathbb{X}_{i,j}(k-1)$,
to $\mathbb{F}_{i}(k)$, and not to $\overline{\mathbb{X}}_{i}(k-1)$.
If source $j$ is not detected again, \emph{i.e.}, $j\notin\mathcal{D}_{i}(k)$ and $j\in \mathcal{K}_{i}(k-1)$, \eqref{eq:Updating solution sets} boils down
to
\begin{equation}
\mathbb{X}_{i,j}(k)=\mathbb{X}_{i,j}(k-1)\setminus\overline{\mathbb{X}}_{i}(k-1).\label{eq:Updating solution sets_v2}
\end{equation}


\subsection{Proving the absence of a source}

\label{subsec:Proving-the-absence}

The intermittent emission of sources makes it difficult to be certain
of the absence of sources in subsets of $\mathbb{X}_{0}$.
A basic solution is to keep UAV~$i$ static
for a duration larger than $\tau_{\max}$. If no emission is perceived,
the corresponding FoD can
be considered free of sources, since the emission period is less than
$\tau_{\max}$ and RF emissions are always detected from sources in the FoD of UAVs. 
This strategy prioritizes the reduction of $\Phi^{\mathrm{a}}_{i}(k)$
in \eqref{eq:uncertainty-estimate} 
over a reduction of $\Phi^{\mathrm{u}}_{i}(k)$ and is not necessarily the best to quickly detect
sources so that UAVs are likely to move continuously. Consequently, some $\boldsymbol{x}\in\mathbb{X}_{0}$
may be observed by various UAVs during non-consecutive time intervals, usually
of duration less than $\tau_{\max}$. Even if the cumulated observation
duration is larger than $\tau_{\max}$, a source may be missed.

Our aim in what follows is to determine under which condition related
to the observation time intervals of a given subset of $\mathbb{X}_{0}$,
this subset may be considered free of sources.

\subsubsection{Time intervals without detection and source detectability}

\label{Ssec:ObservationTimeInt}

Consider a location $\boldsymbol{x}\in\mathbb{X}_{0}$ and $\mathcal{O}_i\left(\boldsymbol{x},k\right)=\left\{ \left[t_{q}\right]\right\} _{q\in\left\llbracket 1,N_{\mathrm{o},i}\left(\boldsymbol{x},k\right)\right\rrbracket }$, the list 
of $N_{\text{o},i}\left(\boldsymbol{x},k\right)$ observation time intervals
during which no emission from $\boldsymbol{x}$ has been detected by UAV~$i$.
$\mathcal{O}_i\left(\boldsymbol{x},k\right)$ is evaluated iteratively as
\begin{equation}
\mathcal{O}_i\left(\boldsymbol{x},k\right)=\mathcal{O}_i\left(\boldsymbol{x},k-1\right)\cup\left\{ \left[\left(k-1\right)T,kT\right]\right\} \label{eq:observation-intervals-list-update}
\end{equation}
for all 
$\boldsymbol{x}\in\mathbb{F}_{i}\left(k\right)\setminus\textstyle\bigcup_{j\in\mathcal{D}_{i}\left(k\right)}\mathbb{X}_{i,j}(k),\label{eq:observation-list-update-condition}$
and 
$\mathcal{O}_i\left(\boldsymbol{x},k\right)=\mathcal{O}_i\left(\boldsymbol{x},k-1\right)\label{eq:observation-list-update-if-detection}$
else. 
Our aim in what follows is to determine using $\mathcal{O}_i\left(\boldsymbol{x},k\right)$
whether there is no source at $\boldsymbol{x}$.


A source $j$ located at $\boldsymbol{x}$ with emission period $\tau_{j}\in[\tau_{\min},\tau_{\max}]$
is \textit{detectable} by UAV~$i$ using $\mathcal{O}_i\left(\boldsymbol{x},k\right)$
if, whatever the initial emission offset $t_{0}\in\left[0,\tau_{j}\right]$,
at least one of the periodic emissions would occur during at least one of the observation time
intervals, \emph{i.e.}, if 
\begin{equation}
\forall t_{0}\in\left[0,\tau_{j}\right[,\exists\ell\in\mathbb{N},\exists\left[t\right]\in\mathcal{O}_i\left(\boldsymbol{x},k\right)\text{ s. t. }t_{0}+\ell\tau_{j}\in\left[t\right],\label{eq:detectable}
\end{equation}
where $[t]=[\underline{t},\overline{t}]$ represents an interval with lower bound $\underline{t}$ and upper bound $\overline{t}$.
\begin{prop}
\label{prop:NoSource}Consider $\boldsymbol{x}\in\mathbb{X}_0$ and $\mathcal{O}_i\left(\boldsymbol{x},k\right)$.
If for all $\tau\in\left[\tau_{\min},\tau_{\max}\right]$, \eqref{eq:detectable}
is satisfied, then no source can be located in $\boldsymbol{x}$.
\end{prop}

\begin{proof}
Assume that a source with emission
period $\tau_{0}$ and offset $t_{0}$ is present at $\boldsymbol{x}$. Since for all $\tau\in\left[\tau_{\min},\tau_{\max}\right]$,
\eqref{eq:detectable} is satisfied, for the period $\tau_{0}$
and offset $t_{0}$, there exists $\ell\in\mathbb{N}$ and $\left[t\right]\in\mathcal{O}_i\left(\boldsymbol{x},k\right)$
such that $t_{0}+\ell\tau_{0}\in\left[t\right]$. Thus, the
source at $\boldsymbol{x}$ has emitted during the time interval
$\left[t\right]$. We assumed that all emissions are detected,
which leads to a contradiction. Consequently, no source is present
at $\boldsymbol{x}$.
\end{proof}

\subsubsection{Proving the absence of a source at a point of $\mathbb{X}_0$}

\label{subsec:Criterion-to-prove}

Proving the absence of sources at some $\boldsymbol{x}$
requires determining whether \eqref{eq:detectable} is satisfied for
all $\tau\in\left[\tau_{\min},\tau_{\max}\right]$. In what follows, we relax this criterion by considering a discrete set of periods $\mathcal{T}_{\mathrm{d}}\subset\left[\tau_{\min},\tau_{\max}\right]$.
This result is extended to a continuous set of emission periods in Appendix \ref{subsec:Extending-criterion-verification}.

Prop.~\ref{prop:ModDiscret} introduces a criterion to verify
whether \eqref{eq:detectable} is satisfied for a given $\tau\in\mathcal{T}_{\mathrm{d}}$.
It involves the modulo of an interval defined as
\begin{equation}
\left[t\right]\mathrm{mod}\tau=\left\{ t\mod{\tau}\mid t\in\left[t\right]\right\} .\label{eq:Interval_Modulo_Def}
\end{equation}
If the width $w\left(\left[t\right]\right)$ of the interval $\left[t\right]$ is such that $w\left(\left[t\right]\right)\geqslant\tau$, \eqref{eq:Interval_Modulo_Def}
is evaluated as $\left[t\right]\mathrm{mod}\tau=\left[0,\tau\right[$
and if $w\left(\left[t\right]\right)<\tau$,
as
\begin{equation}
\left[t\right]\text{mod}\tau=\begin{cases}
\left[\underline{t}\text{mod}\tau,\underline{t}\text{mod}\tau+\overline{t}-\underline{t}\right[ & \text{if }\underline{t}\text{mod}\tau\leqslant\overline{t}\text{mod}\tau,\\
\left[0,\overline{t}\text{mod}\tau\right[\cup\left[\underline{t}\text{mod}\tau,\tau\right[ & \text{else.}
\end{cases}\label{eq:IntervalModulo2}
\end{equation}


\begin{prop}
\label{prop:ModDiscret}Consider $\boldsymbol{x}\in\mathbb{X}_0$, $\mathcal{O}_i\left(\boldsymbol{x},k\right)$, and $\tau\in\mathcal{T}_{\mathrm{d}}$. If
\begin{equation}
\left[0,\tau\right[=\textstyle\bigcup_{\left[t\right]\in\mathcal{O}_i\left(\boldsymbol{x},k\right)}\left[t\right]\mathrm{mod}\tau\label{eq:absence-criterion-verification}
\end{equation}
then \eqref{eq:detectable} is satisfied.
\end{prop}
\begin{proof}
To show that \eqref{eq:absence-criterion-verification} implies \eqref{eq:detectable},
consider $\tau\in\mathcal{T}_{\mathrm{d}}$ with $\tau\neq0$ and
$t_{0}\in\left[0,\tau\right[$. Then \eqref{eq:absence-criterion-verification}
implies that
$\exists\left[t\right]\in\text{\ensuremath{\mathcal{O}}\ensuremath{\left(\boldsymbol{x},k\right)}}$ such that $t_{0}\in\left[t\right]\text{mod}\tau$. Three cases are considered.

First, if $w\left(\left[t\right]\right)\geqslant\tau$, $\left[t\right]\text{mod}\tau=\left[0,\tau\right[$,
so $t_{0}\in\left[0,\tau\right[$. The sequence $u_{\ell}:\ell\mapsto t_{0}+\ell\tau$
is increasing, so $\exists\ell\in\mathbb{N}$ such that
\begin{equation}
t_{0}+\ell\tau\geqslant\underline{t}\geqslant t_{0}+(\ell-1)\tau.\label{eq:proof-2}
\end{equation}
Adding $\tau$ to the second term of \eqref{eq:proof-2}, one gets
$\underline{t}+\tau\geqslant t_{0}+\ell\tau$. As $\overline{t}-\underline{t}\geqslant\tau$,
one has $\overline{t}\geqslant t_{0}+\ell\tau$. Using the first term
of \eqref{eq:proof-2}, one has $\overline{t}\geqslant t_{0}+\ell\tau\geqslant\underline{t}$.
Therefore, $\exists\ell\in\mathbb{N}\text{ such that }t_{0}+\ell\tau\in\left[t\right]$.

Second, if $w\left(\left[t\right]\right)<\tau$ and $\underline{t}\text{mod}\tau\leqslant\overline{t}\text{mod}\tau$,
then $\left[t\right]\text{mod}\tau=\left[\underline{t}\text{mod}\tau,\underline{t}\text{mod}\tau+\overline{t}-\underline{t}\right[$.
As $t_{0}\in\left[t\right]\text{mod}\tau$, $\underline{t}-\underline{\ell}\tau\leqslant t_{0}\leqslant\underline{t}-\underline{\ell}\tau+\overline{t}-\underline{t}$,
with $\underline{\ell}=\left\lfloor \underline{t}/\tau\right\rfloor $.
Thus, $\exists\underline{\ell}\in\mathbb{N}\text{ such that }t_{0}+\underline{\ell}\tau\in\left[t\right]$.

Third, if $w\left(\left[t\right]\right)<\tau$ and $\underline{t}\text{mod}\tau>\overline{t}\text{mod}\tau$
then $\left[t\right]\text{mod}\tau=\left[0,\overline{t}\text{mod}\tau\right[\cup\left[\underline{t}\text{mod}\tau,\tau\right[$.
Two sub-cases are now considered.

First subcase: $t_{0}\in[0,\overline{t}\text{mod}\tau[$
implies $t_{0}\in[0,\overline{t}-\overline{\ell}\tau[$,
with $\overline{\ell}=\lfloor \overline{t}/\tau\rfloor $.
Hence, 
\begin{equation}
\overline{\ell}\tau\leqslant t_{0}+\overline{\ell}\tau\leqslant\overline{t}.\label{eq:Proof1-3-1}
\end{equation}
Moreover, as $\underline{t}\text{mod}\tau>\overline{t}\text{mod}\tau$,
$\exists\underline{\ell}\in\mathbb{N}$ such that $\underline{t}-\underline{\ell}\tau>\overline{t}-\overline{\ell}\tau$.
Hence $\left(\overline{\ell}-\underline{\ell}\right)\tau>\overline{t}-\underline{t}>0$
and $\overline{\ell}>\underline{\ell}$ (because $\tau\geqslant0)$.
This last result gives $\underline{\ell}\tau\leqslant\underline{t}<\overline{\ell}\tau$.
This, with \eqref{eq:Proof1-3-1} yields $\underline{t}\leqslant t_{0}+\overline{\ell}\tau\leqslant\overline{t}$.

Second subcase: $t_{0}\in\left[\underline{t}\text{mod}\tau,\tau\right[$
implies $t_{0}\in\left[\underline{t}-\underline{\ell}\tau,\tau\right[$,
so $\underline{t}\leqslant t_{0}+\underline{\ell}\tau\leqslant\left(\underline{\ell}+1\right)\tau$.
With $\overline{\ell}=\lfloor \overline{t}/\tau\rfloor $
and $\overline{\ell}>\underline{\ell}$, one has $\underline{t}\leqslant t_{0}+\underline{\ell}\tau\leqslant\left(\underline{\ell}+1\right)\tau\leqslant\overline{\ell}\tau\leqslant\overline{t}$.
Thus, in both sub-cases, $\exists\ell\in\mathbb{N}\text{ such that }t_{0}+\ell\tau\in\left[t\right]$.

Finally, in all of the three cases, we have shown that $\exists\ell\in\mathbb{N}\text{ such that }t_{0}+\ell\tau\in\left[t\right]$:
\eqref{eq:detectable} is then verified.
\end{proof}
Proving the absence of sources with emission period in $\mathcal{T}_\mathrm{d}$ at $\boldsymbol{x}$ requires thus verifing \eqref{eq:absence-criterion-verification} for all $\tau\in\mathcal{T}_\mathrm{d}$.
The set $\overline{\mathbb{X}}_{j}(k)$
is then iteratively updated as
\begin{align}
\overline{\mathbb{X}}_{i}(k)=\overline{\mathbb{X}}_{i}(k-1)\cup\{& \boldsymbol{x}\in\mathbb{F}_{i}\left(k\right)\mid\nonumber\\
&\hspace{-1cm}\forall\tau\in\mathcal{T}_{\mathrm{d}},\left[0,\tau\right[=\textstyle\bigcup_{\left[t\right]\in\mathcal{O}_i\left(\boldsymbol{x},k\right)}\left[t\right]\mathrm{mod}\tau\} .\label{eq:Update absence estimate}
\end{align}
This update is only performed for points in $\mathbb{F}_{i}\left(k\right)$
as $\mathcal{O}_i\left(\boldsymbol{x},k\right)$ is only updated for those points using \eqref{eq:observation-intervals-list-update}.
\subsubsection{Proving the absence of sources in subsets of $\mathbb{X}_0$}

\label{subsec:Space-discretization}

If \eqref{eq:absence-criterion-verification} is satisfied for some $\boldsymbol{x}_0\in\mathbb{X}_{0}$ and $\mathcal{O}_i(\boldsymbol{x}_0,k)$, then it also holds true for all $\boldsymbol{x}\in\mathbb{X}_0$ such that $\mathcal{O}_i(\boldsymbol{x},k)=\mathcal{O}_i(\boldsymbol{x}_0,k)$.

To prove the absence of sources in subsets of $\mathbb{X}_0$, $\mathbb{X}_{0}$ is partitioned in boxes $\left[\boldsymbol{x}\right]_{q}$,
$q=1,\dots,N_{\mathrm{c}}$ gathered in a list $\mathcal{X}_{0}$. The lists of observation time intervals are then considered for each box $[\boldsymbol{x}]\in\mathcal{X}_{0}$ and updated as follows. If
$\left[\boldsymbol{x}\right]\subset\mathbb{F}_{i}\left(k\right)\setminus\textstyle\bigcup_{j\in\mathcal{D}_{i}\left(k\right)}\mathbb{X}_{i,j}(k)$,
\begin{equation}
\mathcal{O}_{i}\left(\left[\boldsymbol{x}\right],k\right)=\mathcal{O}_{i}\left(\left[\boldsymbol{x}\right],k-1\right)\cup\left\{ \left[\left(k-1\right)T,kT\right]\right\}. 
\end{equation}
Prop.~\ref{prop:ModDiscret} is easily extended to verify whether \eqref{eq:absence-criterion-verification} is satisfied for all $\boldsymbol{x}\in[\boldsymbol{x}]$ considering $\mathcal{O}_i([\boldsymbol{x}],k)$.
The size of the cells $\left[\boldsymbol{x}\right]\in\mathcal{X}_{0}$
determines the compromise between evaluation complexity and estimation accuracy.
If $\left[\boldsymbol{x}\right]$ is too large, $\mathcal{O}_{i}\left(\left[\boldsymbol{x}\right],k\right)$
may remain empty.


\subsection{Communications and information fusion}

\label{subsec:Communications-1}

To extend the previous framework to a fleet of communicating UAVs, we assume that 
during the time interval $[kT,(k+1)T]$, UAV~$i$ broadcasts $\boldsymbol{x}^{\text{u}}_{i}(k)$,
$\mathcal{K}_{i}(k)$, $\mathcal{X}_{i}(k)$, and $\overline{\mathbb{X}}_{i}(k)$.
The lists $\mathcal{O}_{i}\left(\left[\boldsymbol{x}\right],k\right)$
for all $\left[\boldsymbol{x}\right]\in\mathcal{X}_{0}\setminus\overline{\mathbb{X}}_{i}(k)$ are also transmitted.

After reception of similar information from its neighbors, UAV~$i$ evaluates $\mathcal{K}_{i}(k)=\mathcal{K}_{i}(k)\cup\bigcup_{\ell\in\mathcal{N}_{i}(k)}\mathcal{K}_{\ell}(k)$ and $\overline{\mathbb{X}}_{i}(k)=\overline{\mathbb{X}}_{i}(k)\cup\bigcup_{\ell\in\mathcal{N}_{i}(k)}\overline{\mathbb{X}}_{\ell}(k)$. The set estimates of the source locations for all $j\in\mathcal{K}_{i}(k)$ are updated as
$
\mathbb{X}_{i,j}(k) =\mathbb{X}_{i,j}(k)\cap\underset{\ell\in\mathcal{N}_{i}(k),j\in\mathcal{K}_{\ell}(k)}{\bigcap}\mathbb{X}_{\ell,j}(k).\label{eq:X-update}$

\section{Simulation and results}
\label{sec:Simu-results}




For the numerical simulation, we consider a square search area $\mathbb{X}_{0}$ of $100$~m wide partitioned into
square cells of $2$~m wide in which $10$
sources are randomly placed following a uniform distribution. Each source $j$ emits a periodic signals with period $\tau_{j}$ uniformly distributed in $\left[\tau_{\min},\tau_{\max}\right]=\left[2,16\right]$~s
and initial offset $t_{0,j}$ uniformly distributed in $\left[0,\tau_{j}\right]$.

The initial locations of UAVs are set randomly, following a uniform distribution over $\mathbb{X}_0$. UAVs move at a constant speed and follow a lawn-mower trajectory, see Figure~\ref{fig:Results_0}. 
This enables a fair comparison without control-induced effects. More sophisticated control techniques may be employed as in \cite{Ibenthal2021}. 

\begin{figure}[t!]
\centering
\centerline{\includegraphics[width=8.5cm]{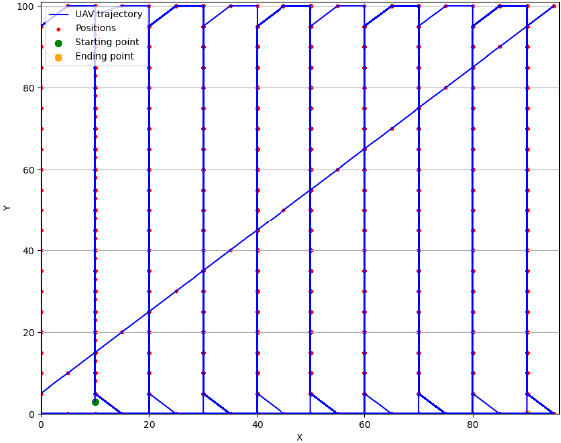}}
\caption{Lawn-mower trajectory followed by UAVs.}
\label{fig:Results_0}
\end{figure}

The FoD $\mathbb{F}_{i}\left(k\right)$
is the intersection of two disks of centers the orthogonal projection of $\boldsymbol{x}^{\mathrm{u}}\left(k-1\right)$ and $\boldsymbol{x}^{\mathrm{u}}\left(k\right)$
on $\mathbb{X}_{0}$ and radius $10$~m. The communication radius
is $r_\mathrm{com}=30$~m, typical of the
Bluetooth Low Energy protocol.
The sampling period for signal acquisition and estimate updates is $T=1$~s.


\subsection{Bayesian baseline algorithm (BBA)}
\label{Ssec:BBA}

An adaptation to several UAVs of the Bayesian estimator introduced in
\cite{Selvaratnam2017} serves as baseline for comparisons. This Bayesian baseline algorithm is denoted BBA. 
During the time interval $\left[\left(k-1\right)T,kT\right]$, the variable $z_{i,j}\left(k\right)=1$ indicates that UAV~$i$ has detected source~$j$ in $\mathbb{F}_{i}\left(k\right)$. If
$z_{i,j}\left(k\right)=0$, the source~$j$ has not been detected.

We assume that the detection probability of a signal emitted by the $j$-th source at $\boldsymbol{x}_j$ is
$\ell_{j}\left(\boldsymbol{x}^{\mathrm{u}}\left(k\right),\boldsymbol{x}_j\right)=
p_{\mathrm{D}}$ if $\boldsymbol{x}_j\in\mathbb{F}_{i}\left(k\right)$ and $\ell_{j}\left(\boldsymbol{x}^{\mathrm{u}}\left(k\right),\boldsymbol{x}_j\right)=
p_{\mathrm{F}}$ if $\boldsymbol{x}_j\notin\mathbb{F}_{i}\left(k\right)$,
where $p_{\mathrm{D}}$ and $p_{\mathrm{F}}$ are the detection and false alarm probabilities. BBA also consider the partition of $\mathbb{X}_{0}$ in the list of boxes $\mathcal{\ensuremath{X}}_{0}$.
For UAV~$i$ located at $\boldsymbol{x}^{\mathrm{u}}_{i}\left(k\right)$, the likelihood of
$z_{i,j}\left(k\right)$ 
for the cell $\left[\boldsymbol{x}\right]$ is taken as the upper-bound of the likelihoods considering all $\boldsymbol{x}\in[\boldsymbol{x}]$
\begin{align}
p\left(z_{i,j}\left(k\right)\mid\boldsymbol{x}^{\mathrm{u}}_{i}\left(k\right),\left[\boldsymbol{x}\right]\right)&=\textstyle\max_{\boldsymbol{x}\in[\boldsymbol{x}]}\ell_{j}\left(\boldsymbol{x}^{\mathrm{u}}_{i}\left(k\right),\boldsymbol{x}\right)^{z_{i,j}\left(k\right)}\nonumber\\
&\left(1-\ell_{j}\left(\boldsymbol{x}^{\mathrm{u}}_{i}\left(k\right),\boldsymbol{x}\right)\right)^{1-z_{i,j}\left(k\right)},\label{eq:observation-proba}
\end{align}
which is easily evaluated by testing whether $[\boldsymbol{x}]$ intersects $\mathbb{F}_{i}\left(k\right)$.

As in \cite{Selvaratnam2017}, UAV~$i$ evaluates $p_{i,j}\left(\left[\boldsymbol{x}\right]\mid z_{i,j}\left(1:k\right),\boldsymbol{x}^{\mathrm{u}}_{i}\left(1:k\right)\right)$, the \emph{a posteriori} probability mass function (pmf) of presence of source~$j$
in a cell $\left[\boldsymbol{x}\right]\in\mathcal{\ensuremath{X}}_{0}$
from the measurements $z_{i,j}\left(1:k\right)$ obtained for the UAV
locations $\boldsymbol{x}^{\mathrm{u}}_{i}\left(1:k\right)$. The number of sources $N^\mathrm{s}$ has to be known and $N^\mathrm{s}$ pmfs are managed in parallel, one for each source, and initialized at $p_{0}={1}/{\left|\mathcal{X}_{0}\right|}$ for all $\left[\boldsymbol{x}\right]\in\mathcal{\ensuremath{X}}_{0}$.

An estimate $\widehat{\boldsymbol{x}}^{\mathrm{t,}\mathrm{B}}_{i,j}\left(k\right)$ of the source location is taken as the mode of $p_{i,j}\left(\left[\boldsymbol{x}\right]\mid z_{i,j}\left(1:k\right),\boldsymbol{x}^{\mathrm{u}}_{i}\left(1:k\right)\right)$ and the estimation uncertainty is the set $\mathbb{X}^{\mathrm{B}}_{i,j}(k)$ of all $\left[\boldsymbol{x}\right]\in\mathcal{\ensuremath{X}}_{0}$ with a pmf larger than some threshold $p^{\mathrm{B}}_{\mathrm{u}}$.   
The set $\overline{\mathbb{X}}^{\mathrm{B}}_{i}(k)$ considered clear of sources in the BBA consists of all $\left[\boldsymbol{x}\right]\in\mathcal{\ensuremath{X}}_{0}$ such that the average over all sources $i$ of $p_{i,j}\left(\left[\boldsymbol{x}\right]\mid z_{i,j}\left(1:k\right),\boldsymbol{x}^{\mathrm{u}}_{i}\left(1:k\right)\right)$ is less than some threshold $p^{\mathrm{B}}_{\mathrm{a}}$.


We choose $p_{\text{F}}=10^{-6}$ and $p_{\mathrm{D}}=1/\tau_\mathrm{avg}$, with $\tau_\mathrm{avg}$ the average of $\tau_{\min}=1~\text{s}$ and $\tau_{\max}=16~\text{s}$, to account the periodicity of emissions. 
The thresholds $p^{\mathrm{B}}_{\mathrm{u}}$ and $p^{\mathrm{B}}_{\mathrm{a}}$ are respectively set to $p^{\mathrm{B}}_{\mathrm{u}}=10^{2}p_{0}$ and $p^{\mathrm{B}}_{\mathrm{a}}=10^{-2}p_{0}$.

Each UAV $i$ updates a probability map $\mathcal{P}_{i,j}\left(k\right)$
for each detected target $j$:
\begin{equation}
\mathcal{P}_{i,j}\left(k\right)=\left\{ p_{i,j}\left(\left[\boldsymbol{x}\right]\mid z_{i,j}\left(1:k\right),\boldsymbol{x}^{\mathrm{u}}_{i}\left(1:k\right)\right)\right\} _{\left[\boldsymbol{x}\right]\in\mathcal{X}_{0}}.\label{eq:probability-map}
\end{equation}
UAVs share their probability maps $\mathcal{P}_{i,j}\left(k\right)$
for all detected targets with their neighbors in list $\mathcal{N}_{i}(k)$.
For the sake of simplicity, let us note $\mathcal{Z}_{i,j}(k)=\left\{ z_{i,j}\left(1:k\right),\boldsymbol{x}^{\mathrm{u}}_{i}\left(1:k\right)\right\} $.
The probability maps of two UAVs $i$ and $\ell$ are the merged using
\begin{equation}
p\left([\boldsymbol{x}]|\mathcal{Z}_{i,j}(k),y_{\ell}\right)=\frac{a}{a+1}\label{eq:fusion-5-1}
\end{equation}
with 
\begin{equation}
a=\frac{p\left([\boldsymbol{x}]|\mathcal{Z}_{i,j}(k)\right)p\left([\boldsymbol{x}]|\mathcal{Z}_{\ell,j}(k)\right)}{\left(1-p\left([\boldsymbol{x}]|\mathcal{Z}_{i,j}(k)\right)\right)\left(1-p\left([\boldsymbol{x}]|\mathcal{Z}_{\ell,j}(k)\right)\right)}\frac{1-p_{0}}{p_{0}},\label{eq:fusion-proba}
\end{equation}
see Appendix \ref{subsec:Appendix---Probability}. This approach assumes
that the observations of UAVs $i$ and $\ell$ are independent, which
may not be always verified in practice.


\subsection{Metrics}

\label{subsec:Metrics}


For UAV~$i$, we consider (\emph{i}) the proportion of detected sources at time $kT$
$d_{i}(k)=\left|\mathcal{K}_{i}(k)\right| / N^{\mathrm{t}},\label{eq:detection rate}$
and (\emph{ii}) once a source has been detected, the average source localization error
\begin{equation}
\varepsilon_{i}\left(k\right)=\frac{1}{\left|\mathcal{K}_{i}(k)\right|}\textstyle\sum_{j\in\mathcal{K}_{i}(k)}\left\Vert \widehat{\boldsymbol{x}}^{\mathrm{t}}_{i,j}\left(k\right)-\boldsymbol{x}^{\mathrm{t}}_{j}\right\Vert _{2},\label{eq:error-definition}
\end{equation}
where $\widehat{\boldsymbol{x}}^{\mathrm{t}}_{i,j}\left(k\right)$
is the estimate of $\boldsymbol{x}^{\mathrm{t}}_{j}$ obtained using SMA (barycenter of the set estimates)
or BBA (mode of the probabilities of presence evaluated for all boxes in $\mathcal{X}_0$),
and (\emph{iii}) the uncertainty on the absence of sources $\Phi^{\mathrm{a}}_{i}(k)$
introduced in \eqref{eq:absence-metric}. 

We evaluate and represent the average of these metrics 
over all UAVs and averages over $10$ independent simulation runs.

\subsection{Results}

\label{subsec:Simulation-results}

Fig.~\ref{fig:Results_0}(a) shows the impact of the speed of UAVs On the proportion of detected targets when 2 UAVs are used. Increasing the exploration speed allows a faster detection of all sources. Fig.~\ref{fig:Results_0}(b) shows the impact of the fleet size when the UAV speed is $1.25\,\text{m/s}$. As expected, the time decreases to localize all targets when the fleet size increases. As the UAV trajectory is the same for SMA and BBA, the curves are the same for both algorithms.

Figs~\ref{fig:Results}(a) and (b) show the evolution of the average source localization error for SMA and BBA, for different UAV speeds with a fleet of 2 UAVs. The SMA provides a reduced error compared to BBA. Reducing the UAV speed reduces the error, as more measurements of the same sources can be acquired during a single sweep.

\begin{figure}[t!]
\begin{minipage}[b]{.49\linewidth}
  \centering
  \centerline{\includegraphics[width=4.5cm]{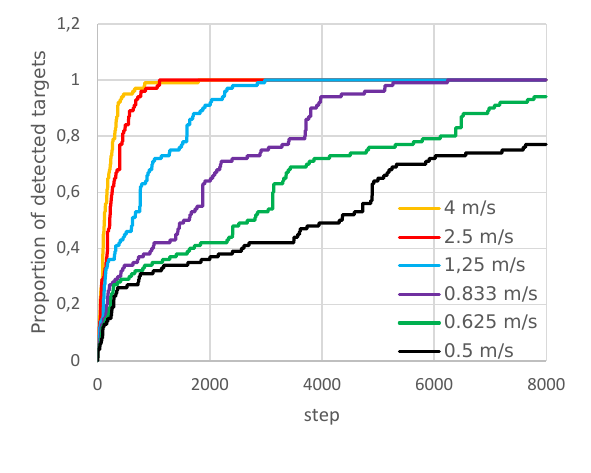}}
\vspace{-0.25cm}
  \centerline{(a)}\medskip
\end{minipage}
\hfill
\begin{minipage}[b]{0.49\linewidth}
  \centering
  \centerline{\includegraphics[width=4.5cm]{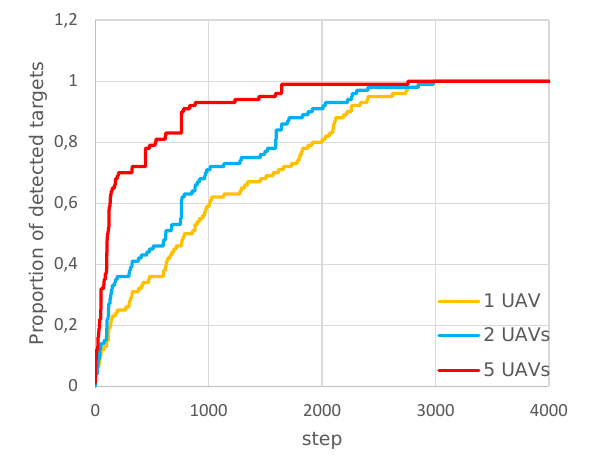}}
\vspace{-0.25cm}
  \centerline{(b)}\medskip
\end{minipage}
\caption{Proportion of detected targets for 2 UAVs with various speeds (a) and for various UAV numbers at a speed of $1.25\,\text{m/s}$ (b).}
\label{fig:Results_0}
\end{figure}

\begin{figure}[t!]

\begin{minipage}[b]{.49\linewidth}
  \centering
  \centerline{\includegraphics[width=4.5cm]{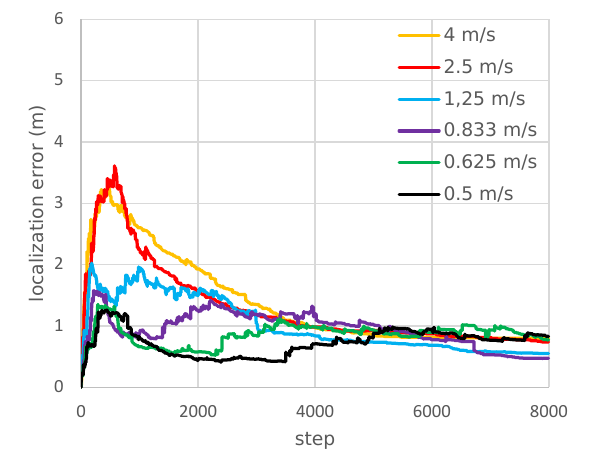}}
\vspace{-0.25cm}
  \centerline{(a)}\medskip
\end{minipage}
\hfill
\begin{minipage}[b]{0.49\linewidth}
  \centering
  \centerline{\includegraphics[width=4.5cm]{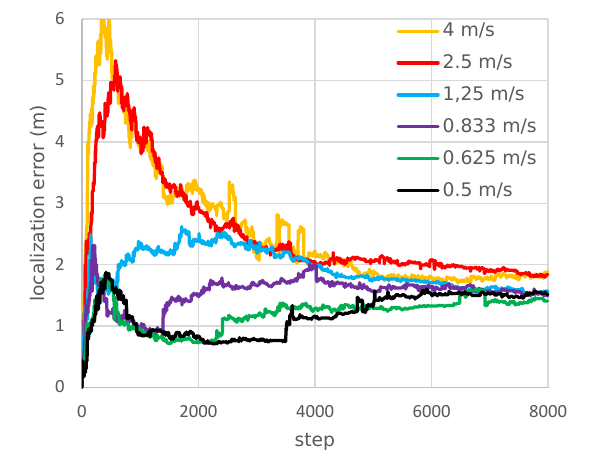}}
\vspace{-0.25cm}
  \centerline{(b)}\medskip
\end{minipage}
\begin{minipage}[b]{.49\linewidth}
  \centering
  \centerline{\includegraphics[width=4.5cm]{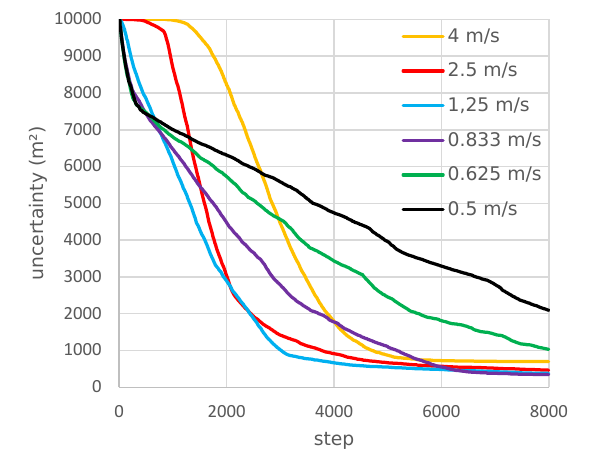}}
\vspace{-0.25cm}
  \centerline{(c)}\medskip
\end{minipage}
\hfill
\begin{minipage}[b]{0.49\linewidth}
  \centering
  \centerline{\includegraphics[width=4.5cm]{{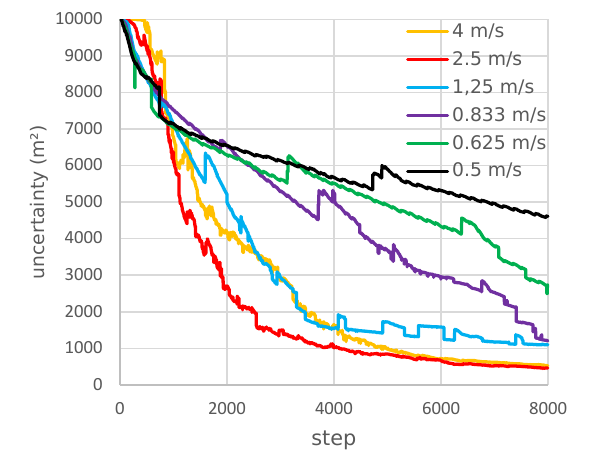}}}
\vspace{-0.25cm}
  \centerline{(d)}\medskip
\end{minipage}
\caption{Average source localization error for various displacement speeds (2 UAVs), SMA (a) and BBA (b); Average size of the set potentially containing sources (2 UAVs), SMA (c) and BBA (d).}
\label{fig:Results}
\end{figure}

\begin{figure}[htpb]

\begin{minipage}[b]{.49\linewidth}
  \centering
  \centerline{\includegraphics[width=4.5cm]{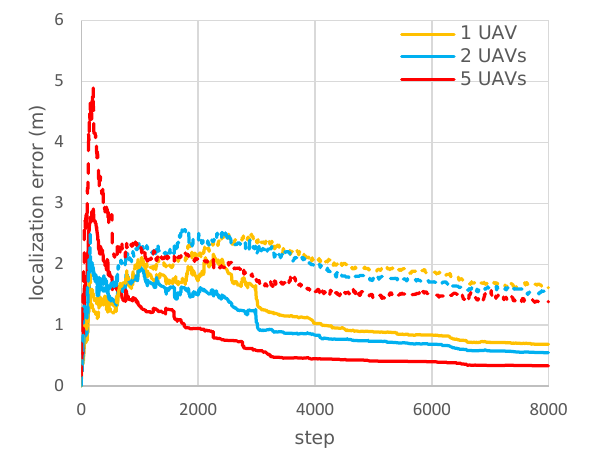}}
\vspace{-0.25cm}
  \centerline{(a)}\medskip
\end{minipage}
\hfill
\begin{minipage}[b]{0.49\linewidth}
  \centering
  \centerline{\includegraphics[width=4.5cm]{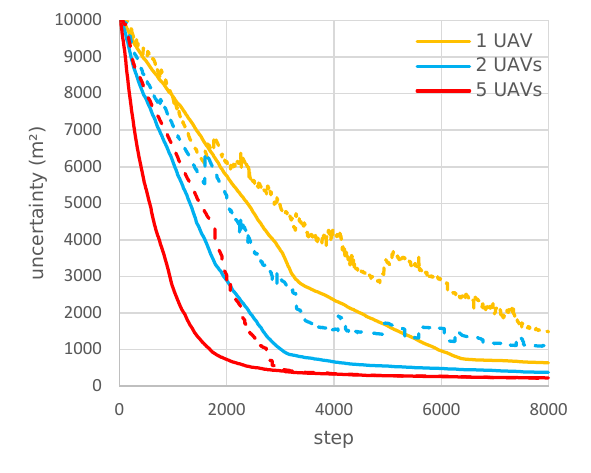}}
\vspace{-0.25cm}
  \centerline{(b)}\medskip
\end{minipage}
\caption{Average error on source localization (a) and Average size of the set potentially containing sources (b) for various fleet sizes evolving at $1.25\,\text{m/s}$, SMA (solid lines) and PBA (dashed lines).}
\label{fig:Results_2}
\end{figure}

Figs~\ref{fig:Results}(c) and (d) show the evolution of the average size of the set potentially containing sources, again for different UAV speeds with a fleet of 2 UAVs for SMA and BBA.
While a displacement at a high speed allows a quick detection and first localization of sources, the best speed to show the absence of targets in large areas are respectively $1.25$ and $2.5\,\text{m/s}$ for the SMA and the BBA. 
The optimal speed results then in a compromise between delay of all source detection, detection accuracy, and ability to show the absence of targets.

Figs~\ref{fig:Results_2} shows the impact of the size of the fleet on the average localization error (a) and on the average size of the set potentially containing sources (b), when the UAV velocity is $1.25\,\text{m/s}$. The simulations confirm that the larger the fleet, the faster all sources are detected and the faster the size of $\mathbb{X}_0 \setminus \overline{\mathbb{X}}_i(k)$ decreases. SMA performs consistently better than BBA.

The evolution of the size of $\overline{\mathbb{X}}^\mathrm{B}_i(k)$ using BBA varies depending on the values of threshold $p^{\mathrm{B}}_{\mathrm{a}}$, but we showed that regardless of the values of this threshold (between $10^{-1}p_{0}$ and $10^{-8}p_{0}$ in the simulation results presented in Figure \ref{fig:uncertainty-thresholds}), SMA performs better than BBA, for the scenario presented in this paper. Indeed, our probabilistic definition of absence uncertainty only quantifies the area for which the probability of containing a target is below an arbitrary threshold \eqref{eq:absence-metric}, set to $p^{\mathrm{B}}_{\mathrm{a}}=p_{0}\cdot10^{-2}=4\cdot10^{-6}$
in the simulations leading to the results presented in this Section
\ref{subsec:Simulation-results}. Consequently, 1) the uncertainty
level highly depends on the chosen threshold (see Figure \ref{fig:uncertainty-thresholds}),
2) this probabilistic uncertainty definition does not have the same
meaning as the deterministic uncertainty, which relies on the certainty
that no target can be located in a certain area (under an assumption
about the maximal emission period).

\begin{figure}[h!]
\centering
\centering \includegraphics[width=6cm]{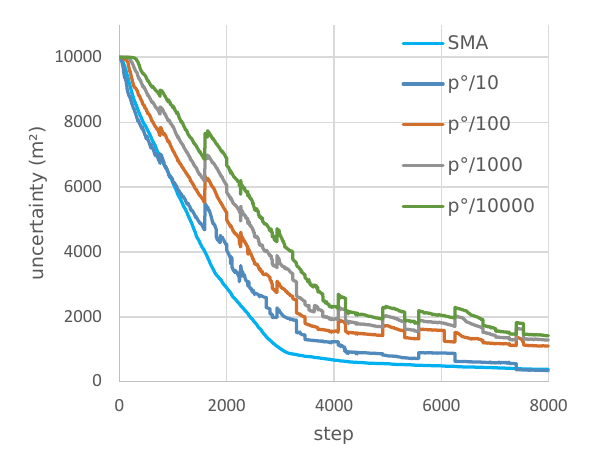}\caption{Average size of the set potentially containing sources
for a fleet of 2 UAVs, comparison between the deterministic algorithm
(SMA) and the Bayesian baseline (PBB) with several values ($10^{-1}p_{0}$,$10^{-2}p_{0}$,
$10^{-3}p_{0}$, $10^{-4}p_{0}$) of threshold $p^{\mathrm{B}}_{\mathrm{a}}$.}
\label{fig:uncertainty-thresholds}
\end{figure}

Videos of the simulation can be found \href{https://centralesupelec-my.sharepoint.com/:f:/g/personal/jacques_bois_centralesupelec_fr/IgCuhnAoY82oQ6oGQEZ7rVxeAS2fYCqGOoajBPiaBqdi_ek?e=EUd63d}{here}.

\section{Conclusion}

\label{sec:Conclusion}

The proposed SMA is able to evaluate set estimates for the location of sources emitting periodically RF signals. It also provides subsets of the search area clear of sources, provided that the hypotheses on the source emission period and the detection probability are satisfied.


Simulations shows that the SMA slightly outperforms
the BBA in terms of localization accuracy and convergence speed. The ability to deterministically localize and certify the absence
of intermittent RF sources has direct applications in search and rescue,
surveillance, electromagnetic intelligence, and connected object monitoring.
By providing mathematically rigorous guarantees, this work offers
a compelling alternative to probabilistic methods in safety-critical
applications.


Among open challenges, one may cite an extension to moving sources with sporadic but non-periodic emissions.

\appendix

\section{Appendix}
\subsection{Extending the criterion of target absence from discrete periods to
period intervals}

\label{subsec:Extending-criterion-verification}

Consider the set of observation time intervals $\mathcal{O}\left(\boldsymbol{x},k\right)$
and assume that for some $\tau\in\mathcal{T}_{\mathrm{d}}$, one has
\begin{equation}
\left[0,\tau\right[=\underset{\left[t\right]\in\mathcal{O}\left(\boldsymbol{x},k\right)}{\bigcup}\left(\left[t\right]\mathrm{mod}\tau\right).\label{eq:criterion}
\end{equation}
Our aim in what follows is to determine an interval $\left[\tau'\right]$
with $\tau\in\left[\tau'\right]$ such that \eqref{eq:criterion}
is satisfied for all $\tau'\in\left[\tau'\right]$. The main idea
is to consider small variations of $\tau$ and to determine conditions
under which \eqref{eq:criterion} is still satisfied. Several cases
depending on $\left|\mathcal{O}\left(\boldsymbol{x},k\right)\right|$
are considered and illustrated in Figure~\ref{fig:Period-verification-extension-scheme}.

\begin{figure}[h!]
\centering
\centering \includegraphics[width=8.5cm]{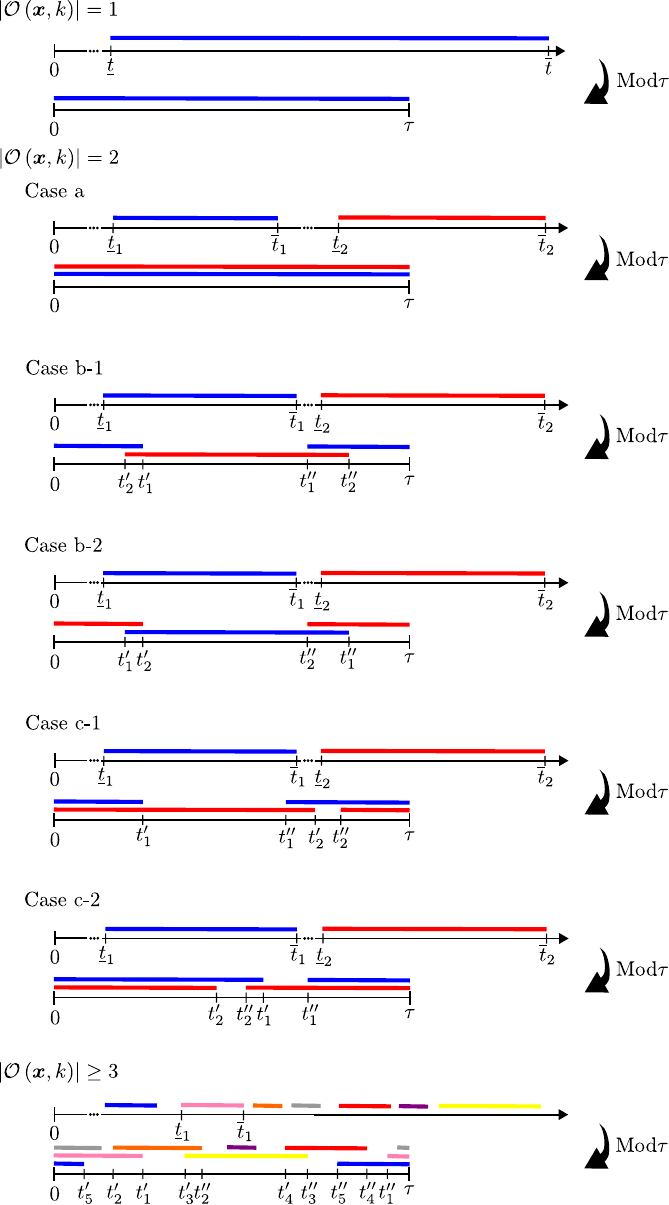}
\caption{Different cases for which $\left[0,\tau\right[=\protect\underset{\left[t\right]\in\mathcal{O}\left(\boldsymbol{x},k\right)}{\bigcup}\left(\left[t\right]\mathrm{mod}\tau\right)$.}
\label{fig:Period-verification-extension-scheme}
\end{figure}

\subsubsection{When $\left|\mathcal{O}\left(\boldsymbol{x},k\right)\right|=1$}

In this case, one has $\mathcal{O}\left(\boldsymbol{x},k\right)=\left\{ \left[t\right]\right\} $.
From \eqref{eq:criterion}, $\left[0,\tau\right[=\left[t\right]\text{mod}\tau$.
With $w\left(\left[t\right]\right)$ the width of interval $\left[t\right]$, for all $\tau'\leqslant w\left(\left[t\right]\right)$,
one has $\left[0,\tau'\right[=\left[t\right]\text{mod}\tau'$. Consequently,
\eqref{eq:criterion} is satisfied for all $\tau'\in\left[0,w\left(\left[t\right]\right)\right]$.
This corresponds to the first case in Figure~\ref{fig:Period-verification-extension-scheme}.

\subsubsection{When $\left|\mathcal{O}\left(\boldsymbol{x},k\right)\right|=2$}

In this case, $\mathcal{O}\left(\boldsymbol{x},k\right)=\left\{ \left[t_{1}\right],\left[t_{2}\right]\right\} $.
Without loss of generality, assume that $\overline{t}_{1}<\underline{t}_{2}$,
\emph{i.e.}, that $\left[t_{2}\right]$ is above $\left[t_{1}\right]$.
\begin{itemize}
\item Case a: If $w\left(\left[t_{1}\right]\right)\geqslant\tau$ or $w\left(\left[t_{2}\right]\right)\geqslant\tau$,
then, a direct application of the result of the case $\left|\text{\ensuremath{\mathcal{O}}}\left(\boldsymbol{x},k\right)\right|=1$
yields that \eqref{eq:criterion} holds true for all $\tau'\in\left[0,\max\left\{ w\left(\left[t_{1}\right]\right),w\left(\left[t_{2}\right]\right)\right\} \right]$.
This corresponds to the second case in Figure~\ref{fig:Period-verification-extension-scheme}.
\item Case b: When $\left[\underline{t}_{1},\overline{t}_{1}\right]\text{mod}\tau$
and $\left[\underline{t}_{2},\overline{t}_{2}\right]\text{mod}\tau$
consist of a different number of intervals, several subcases have
to considered.
\begin{itemize}
\item case b-1: $\left[\underline{t}_{1},\overline{t}_{1}\right] \text{ mod } \tau=\left[0,t_{1}'\right[\cup\left[t_{1}'',\tau\right[$
and $\left[\underline{t}_{2},\overline{t}_{2}\right] \text{ mod } \tau=\left[t_{2}',t_{2}''\right[$,
see the third case in Figure~\ref{fig:Period-verification-extension-scheme}.
Since \eqref{eq:criterion} is satisfied for $\tau,$one has 
\begin{equation}
\begin{array}{c}
t_{1}'\geqslant t_{2}',\\
t_{2}''\geqslant t_{1}''.
\end{array}\label{eq:case-b-1-1}
\end{equation}
Introducing $\ell_{1}=\left\lfloor \overline{t}_{1}/\tau\right\rfloor $,
$\ell_{2}=\left\lfloor \overline{t}_{2}/\tau\right\rfloor $ and
\begin{equation}
\begin{array}{c}
t_{1}'\left(\tau'\right)=\overline{t}_{1}-\ell_{1}\tau',\\
t_{2}'\left(\tau'\right)=\underline{t}_{2}-\ell_{2}\tau',\\
t_{1}"\left(\tau'\right)=\underline{t}_{1}-\left(\ell_{1}-1\right)\tau',\\
t_{2}"\left(\tau'\right)=\overline{t}_{2}-\ell_{2}\tau',
\end{array}\label{eq:case-b-1-3}
\end{equation}
one has $t_{1}'\left(0\right)=t_{1}',\dots,t_{2}"\left(0\right)=t_{2}"$.
Since $\overline{t}_{1}<\underline{t}_{2}$ and $t'_{1}\geqslant t'_{2}$,
one has $\ell_{1}<\ell_{2}$. One has to determine the possible values
of $\tau'$ such that the union of $\left[\underline{t}_{1},\overline{t}_{1}\right]\text{mod}\tau'$
and $\left[\underline{t}_{2},\overline{t}_{2}\right]\text{mod}\tau'$
still yields $\left[0,\tau'\right[$. A sufficient condition is obtained
by considering all $\tau'$ such that
\begin{equation}
\begin{array}{c}
t_{2}'\left(\tau'\right)\leqslant t_{1}'\left(\tau'\right),\\
t_{1}"\left(\tau'\right)\leqslant t_{2}"\left(\tau'\right),\\
t_{2}'\left(\tau'\right)\geqslant0,\\
t_{2}"\left(\tau'\right)\leqslant\tau',
\end{array}\label{eq:case-b-1-4}
\end{equation}
as then, one has 
\begin{align*}
&\left(\left[\underline{t}_{1},\overline{t}_{1}\right]\text{mod}\tau'\right)\cup\left(\left[\underline{t}_{2},\overline{t}_{2}\right]\text{mod}\tau'\right)\\ &=\left[0,t_{1}'\left(\tau'\right)\right[\cup\left[t_{2}'\left(\tau'\right),t_{2}''\left(\tau'\right)\right[\cup\left[t_{1}''\left(\tau'\right),\tau'\right[\label{eq:case-b-1-5}\\
&=\left[0,\tau'\right[.\nonumber 
\end{align*}
Using \eqref{eq:case-b-1-3} in \eqref{eq:case-b-1-4}, one can easily
show that \eqref{eq:criterion} is satisfied for all $\tau'$ in
\begin{equation}
\left[\max\left\{ \frac{\overline{t}_{2}}{\ell_{2}+1},\frac{\underline{t}_{2}-\overline{t}_{1}}{\ell_{2}-\ell_{1}}\right\} ,\min\left\{ \frac{\overline{t}_{2}-\underline{t}_{1}}{\ell_{2}-\ell_{1}+1},\frac{\underline{t}_{2}}{\ell_{2}}\right\} \right].\label{eq:case-b-1-conclusion}
\end{equation}
\item case b-2: $\left[\underline{t}_{1},\overline{t}_{1}\right]\text{mod}\tau=\left[t_{1}',t_{1}''\right[$
and $\left[\underline{t}_{2},\overline{t}_{2}\right]\text{mod}\tau=\left[0,t_{2}'\right[\cup\left[t_{2}'',\tau\right[$,
see the fourth case in Figure~\ref{fig:Period-verification-extension-scheme}.
Since \eqref{eq:criterion} is satisfied for $\tau$, one has 
\begin{equation}
\begin{array}{c}
t_{1}'\leqslant t_{2}'\\
t_{2}''\leqslant t_{1}''
\end{array}.\label{eq:case-b-2-1}
\end{equation}
Consider $\ell_{1}=\left\lfloor \overline{t}_{1}/\tau\right\rfloor $,
$\ell_{2}=\left\lfloor \overline{t}_{2}/\tau\right\rfloor $, and
\begin{equation}
\begin{array}{c}
t_{1}'\left(\tau'\right)=\underline{t}_{1}-\ell_{1}\tau',\\
t_{2}'\left(\tau'\right)=\overline{t}_{2}-\ell_{2}\tau',\\
t_{1}"\left(\tau'\right)=\overline{t}_{1}-\ell_{1}\tau',\\
t_{2}"\left(\tau'\right)=\underline{t}_{2}-\left(\ell_{2}-1\right)\tau'.
\end{array}\label{eq:case-b-2-3}
\end{equation}
Again, one has $t_{1}'\left(0\right)=t_{1}',\dots,t_{2}"\left(0\right)=t_{2}"$.
Now, for all $\tau'$ such that
\begin{equation}
\begin{array}{c}
t_{2}'\left(\tau'\right)\geqslant t_{1}'\left(\tau'\right),\\
t_{1}"\left(\tau'\right)\geqslant t_{2}"\left(\tau'\right),\\
t_{1}'\left(\tau'\right)\geqslant0,\\
t_{1}"\left(\tau'\right)\leqslant\tau',
\end{array}\label{eq:case-b-2-4}
\end{equation}
one has 
\begin{equation}
\left(\left[\underline{t}_{1},\overline{t}_{1}\right]\text{mod}\tau'\right)\cup\left(\left[\underline{t}_{2},\overline{t}_{2}\right]\text{mod}\tau'\right)=\left[0,\tau'\right[.\label{eq:case-b-2-5}
\end{equation}
Since $\overline{t}_{1}<\underline{t}_{2}$ and $t_{1}''\geqslant t_{2}''$,
one has $\ell_{1}+1<\ell_{2}$. Using \eqref{eq:case-b-2-3} in \eqref{eq:case-b-2-4},
one can again easily show that \eqref{eq:criterion} is satisfied
for all $\tau'$ in
\begin{equation}
\left[\max\left(\frac{\overline{t}_{1}}{\ell_{1}+1},\frac{\underline{t}_{2}-\overline{t}_{1}}{\ell_{2}-\ell_{1}-1}\right),\min\left(\frac{\overline{t}_{2}-\underline{t}_{1}}{\ell_{2}-\ell_{1}},\frac{\underline{t}_{1}}{\ell_{1}}\right)\right]
\label{eq:case-b-2-conclusion}
\end{equation}
\end{itemize}
\item case c: When $\left[\underline{t}_{1},\overline{t}_{1}\right]\text{mod}\tau$
and $\left[\underline{t}_{2},\overline{t}_{2}\right]\text{mod}\tau$
both consist of two intervals, one has $\left[\underline{t}_{1},\overline{t}_{1}\right]\text{mod}\tau=\left[0,t_{1}'\right[\cup\left[t_{1}'',\tau\right[$
and $\left[\underline{t}_{2},\overline{t}_{2}\right]\text{mod}\tau=\left[0,t_{2}'\right[\cup\left[t_{2}'',\tau\right[$.
Then, consider $\ell_{1}=\left\lfloor \overline{t}_{1}/\tau\right\rfloor $,
$\ell_{2}=\left\lfloor \overline{t}_{2}/\tau\right\rfloor $ (with
$\ell_{2}>\ell_{1},$else $\left[\underline{t}_{1},\overline{t}_{1}\right]$
and $\left[\underline{t}_{2},\overline{t}_{2}\right]$ would overlap)
and 
\begin{equation}
\begin{array}{c}
t_{1}'\left(\tau'\right)=\overline{t}_{1}-\ell_{1}\tau',\\
t_{2}'\left(\tau'\right)=\overline{t}_{2}-\ell_{2}\tau',\\
t_{1}"\left(\tau'\right)=\underline{t}_{1}-\left(\ell_{1}-1\right)\tau',\\
t_{2}"\left(\tau'\right)=\underline{t}_{2}-\left(\ell_{2}-1\right)\tau',
\end{array}\label{eq:case-c-0-2}
\end{equation}
one has again $t_{1}'\left(0\right)=t_{1}',\dots,t_{2}"\left(0\right)=t_{2}"$.
Two cases have to be considered.
\begin{itemize}
\item case c-1: $t_{2}'>t_{1}''$, see the fifth case in Figure~\ref{fig:Period-verification-extension-scheme}.
Then, for all $\tau'$ such that
\begin{equation}
\begin{array}{c}
t_{2}'\left(\tau'\right)\geqslant t_{1}''\left(\tau'\right),\\
t_{1}'\left(\tau'\right)\geqslant0,\\
t_{2}"\left(\tau'\right)\leqslant\tau',
\end{array}\label{eq:case-c-1-1}
\end{equation}
one has 
\begin{equation}
\left(\left[\underline{t}_{1},\overline{t}_{1}\right]\text{mod}\tau'\right)\cup\left(\left[\underline{t}_{2},\overline{t}_{2}\right]\text{mod}\tau'\right)=\left[0,\tau'\right[.\label{eq:case-c-1-2}
\end{equation}
Introducing \eqref{eq:case-c-0-2} in \eqref{eq:case-c-1-1}, one
easily shows that \eqref{eq:criterion} is satisfied for all 
\begin{equation}
\tau'\in\left[\frac{\underline{t}_{2}}{\ell_{2}},\min\left(\frac{\overline{t}_{1}}{\ell_{1}},\frac{\overline{t}_{2}-\underline{t}_{1}}{\ell_{2}-\ell_{1}+1}\right)\right].\label{eq:case-c-1-conclusion}
\end{equation}
\item case c-2: $t_{1}'>t_{2}''$, see the sixth case in Figure~\ref{fig:Period-verification-extension-scheme}.
Then, for all $\tau'$ such that
\begin{equation}
\begin{array}{c}
t_{1}'\left(\tau'\right)\geqslant t_{2}''\left(\tau'\right),\\
t_{2}'\left(\tau'\right)\geqslant0,\\
t_{1}"\left(\tau'\right)\leqslant\tau',
\end{array}\label{eq:case-c-2-1}
\end{equation}
one has 
\begin{equation}
\left(\left[\underline{t}_{1},\overline{t}_{1}\right]\text{mod}\tau'\right)\cup\left(\left[\underline{t}_{2},\overline{t}_{2}\right]\text{mod}\tau'\right)=\left[0,\tau'\right[.\label{eq:case-c-2-2}
\end{equation}
Introducing \eqref{eq:case-c-0-2} in \eqref{eq:case-c-2-1}, one
shows that \eqref{eq:criterion} is satisfied for all
\begin{equation}
\tau'\in\left[\max\left(\frac{\underline{t}_{1}}{\ell_{1}},\frac{\underline{t}_{2}-\overline{t}_{1}}{\ell_{2}-\ell_{1}-1}\right),\frac{\overline{t}_{2}}{\ell_{2}}\right].\label{eq:case-c-2-conclusion}
\end{equation}
\end{itemize}
\end{itemize}

\subsubsection{When $\left|\mathcal{O}\left(\boldsymbol{x},k\right)\right|\geqslant3$}

In this case, since $\left[0,\tau\right[=\underset{\left[t\right]\in\text{\ensuremath{\mathcal{O}}\ensuremath{\left(\boldsymbol{x},k\right)}}}{\bigcup}\left(\left[t\right]\mathrm{mod}\tau\right)$,
there exists at least one subset $\mathcal{O}'\left(\boldsymbol{x},k\right)$
of $\mathcal{O}\left(\boldsymbol{x},k\right)$ of smallest cardinality
such that 
\begin{equation}
\left[0,\tau\right[=\underset{\left[t\right]\in\mathcal{O}'\left(\boldsymbol{x},k\right)}{\bigcup}\left(\left[t\right]\mathrm{mod}\tau\right)\label{eq:property-subset-O}
\end{equation}
and \eqref{eq:property-subset-O} is not true if a single interval
is removed from $\mathcal{O}'\left(\boldsymbol{x},k\right)$, \textit{i.e.}
$\forall\left[t\right]\in\mathcal{O}'\left(\boldsymbol{x},k\right)$,
\eqref{eq:property-subset-O} is not true for $\mathcal{O}'\left(\boldsymbol{x},k\right)\setminus\left\{ \left[t\right]\right\} $.

Then, one partitions $\mathcal{O}'\left(\boldsymbol{x},k\right)$
as
\begin{equation}
\mathcal{O}'\left(\boldsymbol{x},k\right)=\mathcal{O}'_{1}\left(\boldsymbol{x},k\right)\cup\mathcal{O}'_{2}\left(\boldsymbol{x},k\right)\label{eq:continuous-and-discontinuous}
\end{equation}
where $\mathcal{O}'_{1}\left(\boldsymbol{x},k\right)$ is the list
of intervals $\left[t\right]\in\mathcal{O}'\left(\boldsymbol{x},k\right)$
such that $\left[t\right]\mathrm{mod}\tau$ consists of one interval
and $\mathcal{O}'_{2}\left(\boldsymbol{x},k\right)$ is the list of
intervals $\left[t\right]\in\mathcal{O}'\left(\boldsymbol{x},k\right)$
such that $\left[t\right]\mathrm{mod}\tau$ consists of two intervals.
Figure~\ref{fig:Period-verification-extension-scheme-3} illustrates
a general case with $\left|\mathcal{O}\left(\boldsymbol{x},k\right)\right|\geqslant3$
and \eqref{eq:property-subset-O} satisfied. The gray and purple intervals
are not in $\mathcal{O}'\left(\boldsymbol{x},k\right)$. The orange,
yellow, and red intervals are in $\mathcal{O}'_{1}\left(\boldsymbol{x},k\right)$.
The pink and blue intervals are in $\mathcal{O}'_{2}\left(\boldsymbol{x},k\right)$.

\begin{figure}[h!]
\centering
\centering \includegraphics[width=8.5cm]{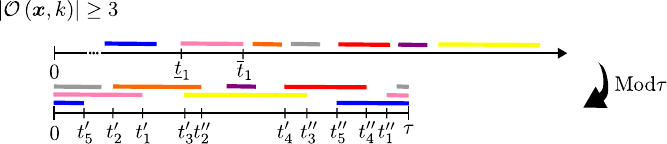}
\caption{General case for $\mathcal{O}\left(\boldsymbol{x},k\right)$ where
$\left[0,\tau\right[=\protect\underset{\left[t\right]\in\mathcal{O}\left(\boldsymbol{x},k\right)}{\bigcup}\left(\left[t\right]\mathrm{mod}\tau\right)$.}
\label{fig:Period-verification-extension-scheme-3}
\end{figure}

\begin{prop}
$\mathcal{O}'_{2}\left(\boldsymbol{x},k\right)$ contains at most
two intervals.
\end{prop}

\begin{proof}
For all $[t_{i}]\in\mathcal{O}'_{2}\left(\boldsymbol{x},k\right)$,
$\left[t_{i}\right]\mathrm{mod}\tau=\left[0,t_{i}'\right[\cup\left[t_{i}'',\tau\right[$.
$\exists\ell\in\left\llbracket 0,\left|\mathcal{O}'_{2}\left(\boldsymbol{x},k\right)\right|\right\rrbracket $
such that 
\begin{equation}
t_{\ell}'=\underset{i\in\left\llbracket 0,\left|\mathcal{O}'_{2}\left(\boldsymbol{x},k\right)\right|\right\rrbracket }{\max}\left\{ t_{i}'\right\} .
\end{equation}
If moreover, 
\begin{equation}
t_{\ell}"=\underset{i\in\left\llbracket 0,\left|\mathcal{O}'_{2}\left(\boldsymbol{x},k\right)\right|\right\rrbracket }{\min}\left\{ t_{i}''\right\} ,
\end{equation}
then 
\begin{equation}
\underset{\left[t\right]\in\mathcal{O}_{2}'\left(\boldsymbol{x},k\right)}{\bigcup}\left(\left[t\right]\mathrm{mod}\tau\right)=\left[t_{\ell}\right]\mathrm{mod}\tau
\end{equation}
and all intervals $\left[t_{i}\right]$ with $i\neq\ell$ cannot belong
to $\mathcal{O}'_{2}\left(\boldsymbol{x},k\right)\subset\mathcal{O}'\left(\boldsymbol{x},k\right)$
since $\mathcal{O}'\left(\boldsymbol{x},k\right)$ is a smallest subset
of $\mathcal{O}\left(\boldsymbol{x},k\right)$ for which \eqref{eq:property-subset-O}
is satisfied.

Now assume that $\exists h\in\left\llbracket 0,\left|\mathcal{O}'_{2}\left(\boldsymbol{x},k\right)\right|\right\rrbracket $,
$h\neq\ell$ such that
\begin{equation}
t_{h}"=\underset{i\in\left\llbracket 0,\left|\mathcal{O}'_{2}\left(\boldsymbol{x},k\right)\right|\right\rrbracket }{\min}\left\{ t_{i}''\right\} ,
\end{equation}
then 
\begin{equation}
\underset{\left[t\right]\in\mathcal{O}_{2}'\left(\boldsymbol{x},k\right)}{\bigcup}\left(\left[t\right]\mathrm{mod}\tau\right)=\left[t_{\ell}\right]\mathrm{mod}\tau\cup\left[t_{h}\right]\mathrm{mod}\tau
\end{equation}
and all intervals $\left[t_{i}\right]$ with $i\neq\ell$ or $i\neq h$
cannot belong to $\mathcal{O}'_{2}\left(\boldsymbol{x},k\right)$.
\end{proof}

In what follows, we consider the case $\left|\mathcal{O}'_{2}\left(\boldsymbol{x},k\right)\right|=2$.
Results are easily extended to the case $\left|\mathcal{O}'_{2}\left(\boldsymbol{x},k\right)\right|=1$.
Let $n'=\left|\text{\ensuremath{\mathcal{O}}'}\left(\boldsymbol{x},k\right)\right|$
and the following indexing of intervals in $\mathcal{O}_{1}'\left(\boldsymbol{x},k\right)$
and $\mathcal{O}_{2}'\left(\boldsymbol{x},k\right)$
\begin{align}
\mathcal{O}'_{1}\left(\boldsymbol{x},k\right) & =\left\{ \left[t_{2}\right],\dots,\left[t_{n'-1}\right]\right\} \\
\mathcal{O}'_{2}\left(\boldsymbol{x},k\right) & =\left\{ \left[t_{1}\right],\left[t_{n'}\right]\right\} 
\end{align}
such that 
\begin{align}
t_{2}' & <\dots<t_{n'-1}'\\
t_{1}' & =\max\left\{ t_{1}',t_{n'}'\right\} 
\end{align}
and
\begin{equation}
t_{n'}"=\min\left\{ t_{1}",t_{n'}"\right\} .
\end{equation}
An example of this indexing is provided in Figure~\ref{fig:Period-verification-extension-scheme-3}.

Consider $\ell_{i}=\left\lfloor \overline{t}_{i}/\tau\right\rfloor $,
$i=1,\dots,n'$ and 
\begin{equation}
\begin{array}{c}
t_{1}'\left(\tau'\right)=\overline{t}_{1}-\ell_{1}\tau',\\
t_{i}'\left(\tau'\right)=\underline{t}_{i}-\ell_{i}\tau',\text{ }i=2,\dots,n'-1\\
t_{n'}'\left(\tau'\right)=\overline{t}_{n'}-\ell_{n'}\tau',\\
t_{1}"\left(\tau'\right)=\underline{t}_{1}-\left(\ell_{1}-1\right)\tau',\\
t_{i}"\left(\tau'\right)=\overline{t}_{i}-\ell_{i}\tau',\text{ }i=2,\dots,n'-1,\\
t_{n'}"\left(\tau'\right)=\underline{t}_{n'}-\left(\ell_{n'}-1\right)\tau'.
\end{array}\label{eq:general-case-1}
\end{equation}
For all $\tau'$ such that
\begin{equation}
\begin{array}{c}
t_{1}'\left(\tau'\right)\geqslant t_{2}'\left(\tau'\right)\\
t_{1}''\left(\tau'\right)\leqslant\tau'\\
t_{i}''\left(\tau'\right)\geqslant t_{i+1}'\left(\tau'\right),\text{ }i=2,\dots,n'-2\\
t_{n'-1}''\left(\tau'\right)\geqslant t_{n'}''\tau'\\
t_{n'}'\left(\tau'\right)\geqslant0
\end{array}\label{eq:extension-general-case-1}
\end{equation}
one has 
\begin{equation}
\left(\left[t_{1}\right]\text{mod}\tau'\right)\cup\dots\cup\left(\left[t_{n'}\right]\text{mod}\tau'\right)=\left[0,\tau'\right[.\label{eq:case-c-1-2-1}
\end{equation}

Combining \eqref{eq:extension-general-case-1} and \eqref{eq:general-case-1},
one gets
\begin{equation}
\begin{array}{c}
\overline{t}_{1}-\underline{t}_{2}\geqslant\left(\ell_{1}-\ell_{2}\right)\tau'\\
\underline{t}_{1}\leqslant\ell_{1}\tau'\\
\overline{t}_{i}-\underline{t}_{i+1}\geqslant\left(\ell_{i}-\ell_{i+1}\right)\tau',\text{ }i=2,\dots,n'-2\\
\overline{t}_{n'-1}-\underline{t}_{n'}\geqslant\left(\ell_{n'-1}-\ell_{n'}+1\right)\tau'\\
\overline{t}_{n'}\geqslant\ell_{n'}\tau'.
\end{array}\label{eq:general-case-2}
\end{equation}

The conditions to be satisfied by $\tau'=\tau'$ to ensure that \eqref{eq:case-c-1-2-1}
holds true depend on the signs of $\ell_{i}-\ell_{i+1}$, $i=1,\dots,n'-2$
and of the sign of $\ell_{n'-1}-\ell_{n'}+1$. Consider the sets of
indices
\[
\mathcal{I}^{-}=\left\{ i=1,\dots,n'-2\mid\ell_{i}<\ell_{i+1}\right\}
\]
and
\[
\mathcal{I}^{+}=\left\{ i=1,\dots,n'-2\mid\ell_{i}>\ell_{i+1}\right\}.
\]
Then, if $\ell_{n'-1}<\ell_{n'}-1$,
\begin{align}
\tau'&\in\left[\max\left(\frac{\underline{t}_{1}}{\ell_{1}},\underset{i\in\mathcal{I}^{-}}{\max}\frac{\overline{t}_{i}-\underline{t}_{i+1}}{\ell_{i}-\ell_{i+1}},\frac{\overline{t}_{n'-1}-\underline{t}_{n'}}{\ell_{n'-1}-\ell_{n'}+1}\right),\right.\nonumber\\
&\hspace{2cm}\left.\min\left(\frac{\overline{t}_{n'}}{\ell_{n'}},\min_{i\in\mathcal{I}^{+}}\frac{\overline{t}_{i}-\underline{t}_{i+1}}{\ell_{i}-\ell_{i+1}}\right)\right],\label{eq:generic-solution}
\end{align}
else, if $\ell_{n'-1}>\ell_{n'}-1$,
\begin{align}
\tau'&\in\left[\max\left(\frac{\underline{t}_{1}}{\ell_{1}},\underset{i\in\mathcal{I}^{-}}{\max}\frac{\overline{t}_{i}-\underline{t}_{i+1}}{\ell_{i}-\ell_{i+1}}\right),\right.\nonumber\\
&\left.\min\left(\frac{\overline{t}_{n'}}{\ell_{n'}},\min_{i\in\mathcal{I}^{+}}\frac{\overline{t}_{i}-\underline{t}_{i+1}}{\ell_{i}-\ell_{i+1}},\frac{\overline{t}_{n'-1}-\underline{t}_{n'}}{\ell_{n'-1}-\ell_{n'}+1}\right)\right].\label{eq:generic-solution-2}
\end{align}

\subsection{Fusion of probability maps between UAVs}

\label{subsec:Appendix---Probability}

With the notations introduced in Section \ref{Ssec:BBA},
this section shows how to update the probability maps \eqref{eq:probability-map}
exchanged between two UAVs. For the sake of simplicity, the following
notations are introduced: $p\left(\left[\boldsymbol{x}\right]\right)$
is the probability that a target is located in cell $\left[\boldsymbol{x}\right]$
and $p\left(\textrm{Ø}\right)=1-p\left(\left[\boldsymbol{x}\right]\right)$
is the probability that the cell $\left[\boldsymbol{x}\right]$ is
empty of targets. Assuming that the observations and trajectories
of UAVs $i$ and $\ell$ are independent,
\begin{equation}
p\left(\mathcal{Z}_{i,j}(k),\mathcal{Z}_{\ell,j}(k)|\left[\boldsymbol{x}\right]\right)=p\left(\mathcal{Z}_{i,j}(k)|\left[\boldsymbol{x}\right]\right)p\left(\mathcal{Z}_{\ell,j}(k)|\left[\boldsymbol{x}\right]\right).\label{eq:fusion-1}
\end{equation}
$p\left(\left[\boldsymbol{x}\right]|\mathcal{Z}_{i,j}(k)\right)$
and $p\left(\left[\boldsymbol{x}\right]|\mathcal{Z}_{\ell,j}(k)\right)$
are known. One has
\begin{align}
p\left(\left[\boldsymbol{x}\right]|\mathcal{Z}_{i,j}(k),\mathcal{Z}_{\ell,j}(k)\right) & =\frac{p\left(\mathcal{Z}_{i,j}(k),\mathcal{Z}_{\ell,j}(k)\mid\left[\boldsymbol{x}\right]\right)}{p\left(\mathcal{Z}_{i,j}(k),\mathcal{Z}_{\ell,j}(k)\right)}p\left(\left[\boldsymbol{x}\right]\right)\label{eq:fusion-2}\\
 & =\frac{p\left(\mathcal{Z}_{i,j}(k)|\left[\boldsymbol{x}\right]\right)p\left(\mathcal{Z}_{\ell,j}(k)|\left[\boldsymbol{x}\right]\right)}{p\left(\mathcal{Z}_{i,j}(k),\mathcal{Z}_{\ell,j}(k)\right)}p\left(\left[\boldsymbol{x}\right]\right)\nonumber 
\end{align}
and
\begin{align}
a&=\frac{p\left(\left[\boldsymbol{x}\right]|\mathcal{Z}_{i,j}(k),\mathcal{Z}_{\ell,j}(k)\right)}{1-p\left(\left[\boldsymbol{x}\right]|\mathcal{Z}_{i,j}(k),\mathcal{Z}_{\ell,j}(k)\right)}\\
& =\frac{p\left(\mathcal{Z}_{i,j}(k)|\left[\boldsymbol{x}\right]\right)p\left(\mathcal{Z}_{\ell,j}(k)|\left[\boldsymbol{x}\right]\right)}{p\left(\mathcal{Z}_{i,j}(k)|\textrm{Ø}\right)p\left(\mathcal{Z}_{\ell,j}(k)|\textrm{Ø}\right)}\frac{p\left(\left[\boldsymbol{x}\right]\right)}{p\left(\textrm{Ø}\right)}\label{eq:fusion-3}\\
 & =\frac{p\left(\left[\boldsymbol{x}\right]\mid\mathcal{Z}_{i,j}(k)\right)p\left(\left[\boldsymbol{x}\right]\mid\mathcal{Z}_{\ell,j}(k)\right)}{p\left(\textrm{Ø}\mid\mathcal{Z}_{i,j}(k)\right)p\left(\textrm{Ø}\mid\mathcal{Z}_{\ell,j}(k)\right)}\frac{p\left(\textrm{Ø}\right)}{p\left(\left[\boldsymbol{x}\right]\right)}\\
 & =\frac{p\left([\boldsymbol{x}]|\mathcal{Z}_{i,j}(k)\right)p\left([\boldsymbol{x}]|\mathcal{Z}_{\ell,j}(k)\right)}{\left(1-p\left([\boldsymbol{x}]|\mathcal{Z}_{i,j}(k)\right)\right)\left(1-p\left([\boldsymbol{x}]|\mathcal{Z}_{\ell,j}(k)\right)\right)}\frac{1-p_{0}}{p_{0}}
\end{align}
thus 
\begin{equation}
p\left(\left[\boldsymbol{x}\right]=1|\mathcal{Z}_{i,j}(k),\mathcal{Z}_{\ell,j}(k)\right)=\frac{a}{a+1}.\label{eq:fusion-5}
\end{equation}

\end{document}